\documentclass[a4paper,UKenglish]{lipics-v2016}

\usepackage[utf8]{inputenc}
\usepackage{amsmath}
\usepackage{amssymb}
\usepackage{amsfonts}
\usepackage{mathpartir}
\usepackage[sort]{cite}
\usepackage{tikz}
\usepackage{paralist}

\title{Rule Formats for Nominal Process Calculi\footnote{Research partially supported by
    the project {Nominal SOS} (nr.~141558-051) of the Icelandic Research Fund, the project
    001-ABEL-CM-2013 within the {NILS} Science and Sustainability Programme, the Spanish
    projects N-Greens Software (S2013/ICE-2731), {TRACES} (TIN2015-67522-C3-3-R),
    StrongSoft (TIN2012-39391-C04) and RISCO (TIN2015-71819-P), and the projects {RACCOON}
    (H2020-EU~714729) and MATHADOR (COGS 724.464) of the European Research Council, and
    the Spanish addition to MATHADOR (TIN2016-81699-ERC).}}

\titlerunning{Rule Formats for Nominal Process Calculi}

\author[1]{Luca Aceto}
\author[1,2,3]{Ignacio F\'abregas}
\author[1,3]{\'Alvaro Garc\'ia-P\'erez}
\author[1]{Anna Ing\'olfsd\'ottir}
\author[2]{Yolanda Ortega-Mall\'en}

\affil[1]{ICE-TCS, School of Computer Science, Reykjavik University (Iceland) \\
  \texttt{\{luca,annai\}@ru.is}}
\affil[2]{Departamento de Sistemas Inform\'aticos y Computaci\'on, Universidad Complutense de Madrid (Spain)\\
  \texttt{\{fabregas,yolanda\}@ucm.es}}
\affil[3]{IMDEA Software Institute, Madrid (Spain) \\
  \texttt{alvaro.garcia.perez@imdea.org}}

\authorrunning{L. Aceto and I. F\'abregas and A. Garc\'ia-P\'erez and A. Ing\'olfsdóttir
  and Y. Ortega-Mall\'en}

\Copyright{L. Aceto and I. F\'abregas and A. Garc\'ia-P\'erez
  and A. Ing\'olfsdóttir and Y. Ortega-Mall\'en}



\EventEditors{Roland Meyer and Uwe Nestmann}
\EventNoEds{2}
\EventLongTitle{28th International Conference on Concurrency Theory (CONCUR 2017)}
\EventShortTitle{CONCUR 2017}
\EventAcronym{CONCUR}
\EventYear{2017}
\EventDate{September 5--8, 2017}
\EventLocation{Berlin, Germany}
\EventLogo{}
\SeriesVolume{85}
\ArticleNo{6} 

\newcommand{\ie}{i.e.}
\newcommand{\eg}{e.g.}
\newcommand{\etc}{etc.}
\newcommand{\etal}{\emph{et al.}}
\newcommand{\Atom}{\mathbb{A}}
\newcommand{\One}{\textup{\bf 1}}
\newcommand{\ASet}{A}
\newcommand{\asort}{\alpha}
\newcommand{\Var}{\mathcal{V}}
\newcommand{\Sta}{S}
\newcommand{\Res}{R}
\newcommand{\Proc}{\mathsf{pr}}
\newcommand{\Act}{\mathsf{ac}}
\newcommand{\Chan}{\mathsf{ch}}
\newcommand{\nTerm}[1]{\mathbb{N}(#1)}
\newcommand{\rTerm}[1]{\mathbb{T}(#1,\Var)}
\newcommand{\gTerm}[1]{\mathbb{T}(#1)}
\newcommand{\Sort}{\mathsf{S}}

\newcommand{\rel}[1]{\buildrel #1\over\longrightarrow}

\newcommand{\nullPA}{\mathit{null}}
\newcommand{\tauPA}[1][]{
	\ifthenelse{\equal{#1}{}}{{\mathit{tau}}}{{\mathit{tau}(#1)}}}
\newcommand{\inPA}[2][]{
	\ifthenelse{\equal{#1}{}}{{\mathit{in}}}{{\mathit{in}(#1,#2)}}}
\newcommand{\outPA}[3][]{
	\ifthenelse{\equal{#1}{}}{{\mathit{out}}}{{\mathit{out}(#1,#2,#3)}}}
\newcommand{\newPA}[1][]{
	\ifthenelse{\equal{#1}{}}{{\mathit{new}}}{{\mathit{new}(#1)}}}
\newcommand{\parPA}[2][]{
	\ifthenelse{\equal{#1}{}}{{\mathit{par}}}{{\mathit{par}(#1,#2)}}}
\newcommand{\sumPA}[2][]{
	\ifthenelse{\equal{#1}{}}{{\mathit{sum}}}{{\mathit{sum}(#1,#2)}}}
\newcommand{\repPA}[1][]{
	\ifthenelse{\equal{#1}{}}{{\mathit{rep}}}{{\mathit{rep}(#1)}}}

\newcommand{\tauAA}{{\mathit{tauA}}}
\newcommand{\inAA}[2][]{
	\ifthenelse{\equal{#1}{}}{{\mathit{inA}}}{{\mathit{inA}(#1,#2)}}}
\newcommand{\outAA}[2][]{
	\ifthenelse{\equal{#1}{}}{{\mathit{outA}}}{{\mathit{outA}(#1,#2)}}}
\newcommand{\boutAA}[2][]{
	\ifthenelse{\equal{#1}{}}{{\mathit{boutA}}}{{\mathit{boutA}(#1,#2)}}}

\newcommand{\resAF}[2]{{(#1,#2)}}
\newcommand{\tr}[2]{(#1\,#2)}
\newcommand{\peract}[2]{#1\cdot #2}
\newcommand{\susp}[2]{#1\bullet #2}
\newcommand{\inj}{\mathrm{inj}}
\newcommand{\fs}{\mathrm{fs}}
\newcommand{\dom}{\mathrm{dom}}

\newcommand{\R}{\mathcal{R}}
\newcommand{\Ru}{\textsc{Ru}}
\newcommand{\NT}[1]{\mathit{NT}[\![#1]\!]}

\newcommand{\fra}{\!\mathrel{\not\!{\not\hspace{-1pt}\approx}}}

\newcommand{\pset}{\mathcal{P}}

\newcommand{\bn}{\mathrm{bn}}
\newcommand{\supp}{\mathrm{supp}}
\newcommand{\perm}[1]{\mathrm{Perm}\;#1}
\newcommand{\permso}[1]{\mathrm{Perm}_s\,#1}
\newcommand{\nf}[1]{\langle #1\rangle_\mathit{nf}}
\newcommand{\SigmaNTS}{\Sigma_{\textup{NTS}}}

\theoremstyle{plain}

\begin{document}

\maketitle

\begin{abstract}
  The nominal transition systems (NTSs) of Parrow et al.\ describe the operational
  semantics of nominal process calculi. We study NTSs in terms of the nominal residual
  transition systems (NRTSs) that we introduce. We provide rule formats for the
  specifications of NRTSs that ensure that the associated NRTS is an NTS and apply them to
  the operational specification of the early pi-calculus. Our study stems from the recent
  Nominal SOS of Cimini et al.\ and from earlier works in nominal sets and nominal logic
  by Gabbay, Pitts and their collaborators.
\end{abstract}

\section{Introduction}
\label{sec:intro}

The goal of this paper is to develop the foundations of a framework for studying the
meta-theory of structural operational semantics (SOS)~\cite{Plo04} for process calculi
with names and name-binding operations, such as the $\pi$-calculi~\cite{SW01}. To this
end, we build on the large body of work on rule formats for SOS, as surveyed
in~\cite{AFV01,MRG07}, and on the nominal techniques of Gabbay, Pitts and their
co-workers~\cite{UPG04,CP07,GM09,Pit13}.

Rule formats provide syntactic templates guaranteeing that the models of the calculi,
whose semantics they specify, enjoy some desirable properties. A first design decision
that has to be taken in developing a theory of rule formats for a class of languages is
therefore the choice of the semantic objects specified by the rules. The target semantic
model we adopt in our study is that of \emph{nominal transition systems} (NTSs), which
have been introduced by Parrow \etal\ in~\cite{PBEGW15,ParrowWBE17} as a uniform model to
describe the operational semantics of a variety of calculi with names and name-binding
operations.  Based on this choice, a basic sanity criterion for a collection of rules
describing the operational semantics of a nominal calculus is that they specify an NTS,
and we present a rule format guaranteeing this property (Thm.~\ref{the:alpha-conversion}).

As a first stepping stone in our study, we introduce \emph{nominal residual transition
  systems} (NRTSs), and study NTSs in terms of NRTSs
(Section~\ref{sec:preliminaries}). More specifically, NRTSs enjoy one desirable property
in the setting of nominal calculi, namely that their transition relation is equivariant
(which means that it treats names uniformly). NTSs are NRTSs that, in addition to having
an equivariant transition relation, satisfy a property Parrow \etal\ call
\emph{alpha-conversion of residuals} (see Def.~\ref{def:nts} for the details). %
The latter property formalises a key aspect of calculi in which names can be scoped to
represent local resources.  To wit, one crucial feature of the $\pi$-calculus is scope
opening~\cite{MPW92}. Consider a transition $p\rel{\overline{a}(\nu b)}p'$ in which a
process $p$ exports a private/local channel name $b$ along channel $a$. Since the name $b$
is local, it `can be subject to alpha-conversion'~\cite{PBEGW15} and the transitions
$p\rel{\overline{a}(\nu c)}p\{b/c\}$ should also be present for each `fresh name' $c$.

In contrast to related work \cite{CMRG12,FG07}, our approach uses nominal terms to connect
the specification system with the semantic model. This has the advantage of capturing the
requirement that transitions be `up to alpha-equivalence' (typical in nominal calculi)
without instrumenting alpha-conversion explicitly in the specification system.

We specify an NRTS by means of a nominal residual transition system specification (NRTSS),
which describes the syntax of a nominal calculus in terms of a nominal signature
(Section~\ref{sec:terms}) and its semantics by means of a set of inference rules
(Section~\ref{sec-specification-nrts}). We develop the basic theory of the NRTS/NRTSS
framework, building on the nominal algebraic datatypes of Pitts~\cite{Pit13} and the
nominal rewriting framework of Fernandez and Gabbay~\cite{FG07}. Based on this framework,
we provide rule formats \cite{AFV01,MRG07} for NRTSSs (Section~\ref{sec-rule-format-nrts})
that ensure that the induced transition relation is equivariant
(Thm.~\ref{thm:rule-format-equivariance}) and enjoys alpha-conversion of residuals
(Thm.~\ref{the:alpha-conversion}), and is therefore an NTS.
Section~\ref{sec:early-pi-calculus} presents an example of application of our rule formats
to the setting of the $\pi$-calculus, and Section~\ref{sec:conclusions} discusses avenues
for future work, as well as related work, and concludes.


\section{Preliminaries}
\label{sec:preliminaries}

\paragraph*{Nominal sets}

We follow earlier foundational work by Gabbay and Pitts on nominal sets in
\cite{GP02,Pit13,Pit16}. We assume a countably infinite set $\Atom$ of \emph{atoms} and
consider $\perm{\Atom}$ as the group of \emph{finite permutations of atoms} (hereafter
\emph{permutations}) ranged over by $\pi$, where we write $\iota$ for the \emph{identity},
$\circ$ for \emph{composition} and $\pi^{-1}$ for the \emph{inverse} of permutation
$\pi$. We are particularly interested in \emph{transpositions} of two atoms: $\tr{a}{b}$
stands for the permutation that swaps $a$ with $b$ and leaves all other atoms fixed. Every
permutation $\pi$ is equal to the composition of a finite number of transpositions, \ie\
$\pi=\tr{a_1}{b_1}\circ\ldots\circ\tr{a_n}{b_n}$ with $n\geq 0$.

An \emph{action} of the group $\perm{\Atom}$ on a set $S$ is a binary operation mapping
each $\pi\in\perm{\Atom}$ and $s\in S$ to an element $\peract{\pi}{s}\in S$, and
satisfying the identity law $\peract{\iota}{s}=s$ and the composition law
$\peract{(\pi_1\circ\pi_2)}{s}=\peract{\pi_1}{(\peract{\pi_2}{s})}$. A
\emph{$\perm{\Atom}$-set} is a set equipped with an action of $\perm{\Atom}$.

We say that a set of atoms $A$ \emph{supports} an object $s$ iff $\peract{\pi}{s}=s$ for
every permutation $\pi$ that leaves each element $a\in A$ invariant. In particular, we are
interested in sets all of whose elements have finite support (Def.~2.2 of \cite{Pit13}).

\begin{definition}[Nominal sets] A \emph{nominal set} is a $\perm{\Atom}$-set all of whose
	elements are finitely supported.
\end{definition}

For each element $s$ of a nominal set, we write $\supp(s)$ for the least set that supports
$s$, called the \emph{support} of $s$. (Intuitively, the action of permutations on a set
$S$ determines that a finitely supported $s\in S$ only depends on atoms in $\supp(s)$, and
no others.)  The set $\Atom$ of atoms is a nominal set when $\peract{\pi}{a}=\pi a$ since
$\supp(a)=\{a\}$ for each atom $a\in\Atom$. The set $\perm{\Atom}$ of finite permutations
is also a nominal set where the permutation action on permutations is given by
conjugation, \ie\ $\peract{\pi}{\pi'}=\pi\circ\pi'\circ\pi^{-1}$, and the support of a
permutation $\pi$ is $\supp(\pi)=\{a\mid \pi a\not=a\}$.

Given two $\perm{\Atom}$-sets $S$ and $T$ and a function $f:S\to T$, the action of
permutation $\pi$ on function $f$ is given by conjugation, \ie\
$(\peract{\pi}{f})(s)=\peract{\pi}{f(\peract{\pi^{-1}}{s})}$ for each $s\in S$. We say
that a function $f:S\to T$ is \emph{equivariant} iff
$\peract{\pi}{f(s)}=f(\peract{\pi}{s})$ for every $\pi\in\perm{\Atom}$ and every $s\in
S$.
The intuition is that an equivariant function $f$ is atom-blind, in that $f$ does not
treat any atom preferentially. It turns out that a function $f$ is equivariant iff
$\supp(f)=\emptyset$ (Rem.~2.13 of \cite{Pit13}). The function $\supp$ is equivariant
(Prop.~2.11 of \cite{Pit13}).

Let $S$ be a $\perm{\Atom}$-set, we write $S_\fs$ for the nominal set that contains the
elements in $S$ that are finitely supported. The \emph{nominal function set} between
nominal sets $S$ and $T$ is the nominal set $S\to_\fs T$ of finitely supported functions
from $S$ to $T$---be they equivariant or not. Let $S_1$ and $S_2$ be nominal sets. The
product $S_1\times S_2$ is a nominal set (Prop.~2.14 of \cite{Pit13}). The permutation
action for products is given componentwise (Eq~(1.12) of \cite{Pit13}).

An element $s_1\in S_1$ \emph{is fresh in}
$s_2\in S_2$, written $s_1\#s_2$, iff $\supp(s_1)\cap\supp(s_2)=\emptyset$. The freshness
relation is equivariant (Eq.~(3.2) of \cite{Pit13}).

Finally, we consider \emph{atom abstractions} (Sec.~4 of \cite{Pit13}), which represent
alpha-equivalence classes of elements.

\begin{definition}[Atom abstraction]
  \label{def:atom-abstraction}
  Given a nominal set $S$, the \emph{atom abstraction} of atom $a$ in element $s\in S$,
  written $\langle a\rangle s$, is the $\perm{\Atom}$-set
  $\langle a\rangle s=\{(b,\peract{\tr{b}{a}}{s})\mid b=a\lor b\# s\}$, whose permutation
  action is
  $ \peract{\pi}{\langle a\rangle s}=\{
  (\peract{\pi}{b},\peract{\pi}{(\peract{\tr{b}{a}}{s})})\mid
  \peract{\pi}{b}=\peract{\pi}{a}\lor \peract{\pi}{b}\# \peract{\pi}{s}\}$.

  We write $[\Atom]S$ for the set of \emph{atom abstractions} in elements of $S$, which is
  a nominal set (Def.~4.4 of \cite{Pit13}), since
  $\supp(\langle a\rangle s)=\supp(s)\setminus\{a\}$ for each atom $a$ and element
  $s\in S$.
\end{definition}

\paragraph*{Nominal Transition Systems}
Nominal transitions systems adopt the state/residual presentation for transitions of
\cite{BP09}, where a residual is a pair consisting of an action and a state. In
\cite{PBEGW15}, Parrow \etal\ develop modal logics for process algebras à la
Hennessy-Milner. Here we are mainly interested in the transition relation and we adapt
Definition~1 in \cite{PBEGW15} by removing the predicates. We write $\pset_\omega(\Atom)$
for the \emph{finite power set} of $\Atom$.

\begin{definition}[Nominal transition system]
  \label{def:nts}
  A \emph{nominal transition system} (NTS) is a quadruple $(S,\mathit{Act},\bn,\rel{})$
  where $S$ and $\mathit{Act}$ are nominal sets of \emph{states} and \emph{actions}
  respectively, $\bn:\mathit{Act}\to\pset_\omega(\Atom)$ is an equivariant function that
  delivers the \emph{binding names} in an action, and
  ${\rel{}}\subseteq S\times(\mathit{Act}\times S)$ is an equivariant binary transition
  relation from states to \emph{residuals} (we let $\mathit{Act}\times S$ be the set of
  residuals). The function $\bn$ is such that $\bn(\ell)\subseteq\supp(\ell)$ for each
  $\ell\in \mathit{Act}$. We often write $p\rel{}(\ell,p')$ in lieu of
  $(p,(\ell,p'))\in{\rel{}}$.

  Finally, the transition relation $\rel{}$ must satisfy \emph{alpha-conversion of
    residuals}, that is, if ${a\in \bn(\ell)}$, $b\#(\ell,p')$ and $p\rel{}(\ell,p')$ then
  also $p\rel{}(\peract{\tr{a}{b}}{\ell},\peract{\tr{a}{b}}{p'})$, or equivalently
  $p\rel{}\peract{\tr{a}{b}}{(\ell,p')}$.
\end{definition}

We will consider an NTS (without its associated binding-names function $\bn$) as a
particular case of a nominal residual transition system, which we introduce next.
\begin{definition}[Nominal residual transition system]
  \label{def:nrts}
  A \emph{nominal residual transition system} (NRTS) is a triple $(\Sta,\Res,\rel{})$
  where $\Sta$ and $\Res$ are nominal sets, and where ${\rel{}}\subseteq S\times R$ is an
  equivariant binary transition relation. We say $\Sta$ is the set of \emph{states} and
  $\Res$ is the set of \emph{residuals}.
\end{definition}
The connection between NTSs and NRTSs will be explained in more detail in
Section~\ref{sec-rule-format-nrts}.

\section{Nominal terms}
\label{sec:terms}

This section is devoted to the notion of nominal terms, which are syntactic objects that
make use of the atom abstractions of Definition~\ref{def:atom-abstraction} and represent
terms up to alpha-equivalence. As a first step, we introduce raw terms, devoid of any notion of alpha-equivalence.
Our raw terms resemble those from the literature,
mainly \cite{Pit13,UPG04,FG07,CP07}, but with some important differences. In particular,
our terms include both variables (\ie\ unknowns) and moderated terms (\ie\ explicit
permutations over raw terms), and we consider atom and abstraction sorts. (The raw terms
of \cite{Pit13} do not include moderated terms, and the ones in \cite{UPG04,FG07} only
consider moderated variables. In \cite{CP07} the authors consider neither atom nor
abstraction sorts.) We also adopt the classic presentation of free algebras and term
algebras in \cite{GTWW77,BS00} in a different way from that in \cite{CP07,Pit13}. The raw
terms correspond to the standard notion of free algebra over a signature generated by a
set of variables. We then adapt the $\Sigma$-structures of \cite{CP07} to our sorting
schema. Finally, the nominal terms are the interpretations of the ground terms in the
initial $\Sigma$-structure; they coincide with the nominal algebraic
terms of \cite{Pit13}.

\begin{definition}[Nominal signature and nominal sort]
  A \emph{nominal signature} (or simply a \emph{signature}) $\Sigma$ is a triple
  $(\Delta,\ASet,F)$ where $\Delta=\{\delta_1,\ldots,\delta_n\}$ is a finite set of
  \emph{base sorts}, $\ASet$ is a countable set of \emph{atom sorts}, and $F$ is a finite
  set of \emph{function symbols}. The \emph{nominal sorts} over $\Delta$ and $\ASet$ are
  given by the grammar
  $\sigma~::=~\delta\mid \asort \mid [\asort]\sigma\mid \sigma_1 \times \ldots \times
  \sigma_k,$
  with $k\geq 0$, $\delta\in\Delta$ and $\asort\in\ASet$. The sort $[\asort]\sigma$ is the
  \emph{abstraction sort}. Symbol $\times$ denotes the \emph{product sort}, which is
  associative; $\sigma_1\times\ldots\times\sigma_k$ stands for the sort of the empty
  product when $k=0$, which we may write as $\One$. We write $\Sort$ for the set of
  nominal sorts. We arrange the function symbols in $F$ based on the sort of the data that
  they produce. We write $f_{ij}\in F$ with $1\leq i\leq n$ and $1\leq j \leq m_i$ such
  that $f_{ij}$ has arity $\sigma_{ij}\to\delta_i$, where $\delta_i$ is a base sort.
\end{definition}

The theory of nominal sets extends to the case of (countably) many-sorted atoms (see
Sec.~4.7 in \cite{Pit13}). We assume that $\Atom$ contains a countably infinite
collection of atoms $a_\asort$, $b_\asort$, $c_\asort$,~\ldots\ for each atom sort
$\asort$ such that the sets of atoms $\Atom_\asort$ of each sort are mutually disjoint. We
write
$\permso{\Atom}=\{\pi\in\perm{\Atom}\mid \forall \asort\in\ASet.\, \forall
a\in\Atom_\asort.\;\pi\,a\in\Atom_\asort\}$
for the subgroup of finite permutations that respect the sorting. The sorted nominal sets
are the $\permso{\Atom}$-sets whose elements are finitely supported. In the sequel we may
drop the $s$ subscript in $\permso{\Atom}$ and omit the `sorted' epithet from `sorted
nominal sets'.

We let $\Var$ be a set that contains a countably infinite collection of \emph{variable
  names} (variables for short) $x_\sigma$, $y_\sigma$, $z_\sigma$, \ldots\ for each sort
$\sigma$, such that the sets of variables $\Var_\sigma$ of each sort are mutually
disjoint. We also assume that $\Var$ is disjoint from $\Atom$.

\begin{definition}[Raw terms]
  \label{def:raw-terms}
  Let $\Sigma=(\Delta,A,F)$ be a signature. The set of \emph{raw terms over signature
    $\Sigma$ and set of variables $\Var$} (\emph{raw terms} for short) is given by the
  grammar
  \begin{displaymath}
    \begin{array}{rcl}
      t_\sigma ::=&
                    x_\sigma\mid
                    a_\asort\mid
                    (\susp{\pi}{t_\sigma})_\sigma\mid
                    ([a_\asort]t_\sigma)_{[\asort]\sigma}\mid
                    (t_{\sigma_1},\ldots,t_{\sigma_k})_{\sigma_1\times\ldots\times\sigma_k} \mid
                    (f_{ij}(t_{\sigma_{ij}}))_{\delta_i},
    \end{array}
  \end{displaymath}
  where term $x_\sigma$ is a \emph{variable} of sort $\sigma$, term $a_\alpha$ is an
  \emph{atom} of sort $\alpha$, term $(\susp{\pi}{t_\sigma})_\sigma$ is a \emph{moderated
    term} (\ie\ the explicit, or delayed, permutation $\pi$ over term $t_\sigma$), term
  $([a_\asort]t_\sigma)_{[\asort]\sigma}$ is the \emph{abstraction of atom $a_\asort$ in
    term $t_\sigma$}, term
  $(t_{\sigma_1},\ldots,t_{\sigma_k})_{\sigma_1\times\ldots\times\sigma_k}$ is the
  \emph{product of terms} $t_{\sigma_1}$, \ldots, $t_{\sigma_k}$, and term
  $(f_{ij}(t_{\sigma_{ij}}))_{\delta_i}$ is the \emph{datum of base sort $\delta_i$
    constructed from term $t_{\sigma_{ij}}$ and function symbol
    $f_{ij} : \sigma_{ij}\to\delta_i$}.  When they are clear from the context or
  immaterial, we leave the arities and sorts implicit and write $x$, $a$, $\susp{\pi}{t}$,
  $[a]t$, $(t_1,\ldots,t_k)$, $f(t)$, \etc
\end{definition}

The raw terms are the inhabitants of the carrier of the free algebra over the set of
variables $\Var$ and over the $\Sort$-sorted conventional signature that consists of the
function symbols in $F$, together with a constant symbol for each atom $a_\alpha$, a unary
symbol that produces moderated terms for each permutation $\pi$ and each sort $\sigma$, a
unary symbol that produces abstractions for each atom $a_\alpha$ and sort $\sigma$, and a
$k$-ary symbol that produces a product of sort $\sigma_1\times\ldots\times\sigma_k$ for
each sequence of sorts $\sigma_1$, \ldots, $\sigma_k$. (See \cite{GTWW77} for a classic
presentation of term algebras, initial algebra semantics and free algebras.)

We write $\rTerm{\Sigma}_\sigma$ for the set of raw terms of sort $\sigma$. A raw term $t$
is \emph{ground} iff no variables occur in $t$. We write $\gTerm{\Sigma}_\sigma$ for the
set of ground terms of sort $\sigma$. The sets of raw terms (resp.\ ground terms) of each
sort are mutually disjoint as terms carry sort information. Therefore we sometimes
identify the family $(\rTerm{\Sigma}_\sigma)_{\sigma\in\Sort}$ of $\Sort$-indexed raw
terms and the family $(\gTerm{\Sigma}_\sigma)_{\sigma\in\Sort}$ of $\Sort$-indexed ground
terms with their respective ranges $\bigcup_{\sigma\in\Sort}\rTerm{\Sigma}_\sigma$ and
$\bigcup_{\sigma\in\Sort}\gTerm{\Sigma}_\sigma$, which we abbreviate as $\rTerm{\Sigma}$
and $\gTerm{\Sigma}$ respectively.

The set $\rTerm{\Sigma}$ of raw terms is a nominal set, with the $\perm{\Atom}$-action and
the support of a raw term given by:
\begin{displaymath}
\begin{array}{cc}\hspace{-.3cm}
\begin{array}[t]{rcl}
\peract{\pi}{x}&=&x\\
\peract{\pi}{a}&=&\pi\,a\\
\peract{\pi}{(\susp{\pi_1}{t})}&=&\susp{(\peract{\pi}{\pi_1})}
{(\peract{\pi}{t})}\\
\peract{\pi}{[a]t}&=&[\pi\,a](\peract{\pi}{t})\\
\peract{\pi}{(t_1,\ldots,t_k)}&=&
(\peract{\pi}{t_1},\ldots,\peract{\pi}{t_k})\\
\peract{\pi}{(f(t))}&=&
f(\peract{\pi}{t}),\\[8pt]
\end{array}&
\begin{array}[t]{rcl}
\supp(x)&=&\emptyset\\
\supp(a)&=&\{a\}\\
\supp(\susp{\pi}{t})&=&\supp(\pi)\cup\supp(t)\\
\supp([a](t))&=&\{a\}\cup\supp(t)\\
\supp((t_1,\ldots,t_k))&=&\supp(t_1)\cup\ldots\cup\supp(t_k)\\
\supp(f(t))&=&\supp(t).
\end{array}
\end{array}
\end{displaymath}

It is straightforward to check that the permutation action for raw terms is
sort-preserving (remember that permutations are also sort-preserving). The set $\gTerm{\Sigma}$ of ground terms is also a nominal set since it is closed with respect to the $\perm{\Atom}$-action given above.

\begin{example}[$\pi$-calculus]
  \label{ex:pi-calculus}
  Consider a signature $\Sigma$ for the $\pi$-calculus \cite{SW01,CMRG12} given by a
  single atom sort $\Chan$ of channel names, and base sorts $\Proc$ and $\Act$ for
  processes and actions respectively. The function symbols (adapted from \cite{SW01}) are
  the following:
  \begin{displaymath}
    \begin{array}{lll}
      F=\{
      \begin{array}[t]{l}
	\nullPA:\One\to \Proc,\\
	\tauPA[]:\Proc\to \Proc,\\
	\inPA[]{}:(\Chan\times [\Chan]\Proc)\to \Proc,\\
	\outPA{}{}{}:(\Chan\times \Chan\times \Proc)\to
	\Proc,
      \end{array}
      \begin{array}[t]{l}
	\parPA[]{}:(\Proc\times \Proc)\to \Proc,\\
	\sumPA[]{}:(\Proc\times \Proc)\to \Proc,\\
	\repPA[]:\Proc\to \Proc,\\
	\newPA[]:[\Chan]\Proc\to \Proc,
      \end{array}
      \begin{array}[t]{l}
	\tauAA:\One\to\Act,\\
	\inAA[]{}:(\Chan\times \Chan)\to\Act,\\
	\outAA[]{}:(\Chan\times \Chan)\to\Act,\\
	\boutAA[]{}:(\Chan\times \Chan)\to\Act\quad\}.
      \end{array}
    \end{array}
  \end{displaymath}

  Recalling terminology from \cite{SW01,CMRG12}, $\nullPA$ stands for inaction,
  $\tauPA[p]$ for the internal action after which process $p$ follows, $\inPA[a]{[b]p}$
  for the input at channel $a$ where the input name is bound to $b$ in the process $p$
  that follows, $\outPA[a]{b}{p}$ for the output of name $b$ through channel $a$ after
  which process $p$ follows, $\parPA[p]{q}$ for parallel composition, $\sumPA[p]{q}$ for
  nondeterministic choice, $\repPA[p]$ for parallel replication, and $\newPA[{[a]p}]$ for
  the restriction of channel $a$ in process $p$ ($a$ is private in $p$). Actions and
  processes belong to different sorts. We use $\tauAA$, $\outAA[a]{b}$, $\inAA[a]{b}$ and
  $\boutAA[a]{b}$ respectively for the internal action, the output action, the input
  action and the bound output action.

  The set of terms of the $\pi$-calculus corresponds to the subset of ground terms over
  $\Sigma$ of sort $\Proc$ and $\Act$ in which no moderated (sub-)terms occur. For
  instance, the process $(\nu b)(\overline{a}b.0)$ corresponds to the ground term
  $\newPA[{[b](\outPA[a]{b}{\nullPA})}]$, whose support is $\{a,b\}$. Both free and bound
  channel names (such as the $a$ and $b$ respectively in the example process) are
  represented by atoms. The set of ground terms also contains generalised processes and
  actions with moderated (sub-)terms $\susp{\pi}{p}$, which stand for a delayed
  permutation $\pi$ that ought to be applied to a term $p$, \eg\
  $\newPA[\susp{\pi}{([b](\outPA[a]{b}{\nullPA}))}]$.\qed
\end{example}

Raw terms allow variables to occur in the place of any ground subterm. The variables
represent \emph{unknowns}, and should be mistaken neither with free nor bound channel
names. For instance, the raw term $\newPA[{[b](\outPA[a]{b}{x})}]$ represents a process
$(\nu b)(\overline{a}b.P)$ where the $x$ is akin to the meta-variable $P$, which stands
for some unknown process. The process $(\nu b)(\overline{a}b.P)$ unifies with
$(\nu b)(\overline{a}b.0)$ by replacing $P$ with $0$. In the nominal setting, the raw term
$\newPA[{[b](\outPA[a]{b}{x})}]$ unifies with ground term
$\newPA[{[b](\outPA[a]{b}{\nullPA})}]$, by means of a \emph{substitution} $\varphi$ such
that $\varphi(x)=\nullPA$. Formally, substitutions are defined below.

\begin{definition}[Substitution]
  \label{def:substitution}
  A \emph{substitution} $\varphi:\Var\to_\fs\rTerm{\Sigma}$ is a sort-preserving, finitely
  supported function from variables to raw terms. The \emph{domain $\dom(\varphi)$ of a
    substitution $\varphi$} is the set ${\{x\mid \varphi(x)\not=x\}}$. A substitution
  $\varphi$ is \emph{ground} iff $\varphi(x)\in\gTerm{\Sigma}$ for every variable
  $x\in\dom(\varphi)$.
\end{definition}

The set of substitutions is a nominal set. The \emph{extension to raw terms
  $\overline{\varphi}$ of substitution $\varphi$} is the unique homomorphism induced by
$\varphi$ from the free algebra $\rTerm{\Sigma}$ to itself, which coincides with the
function given by:
\begin{displaymath}
\begin{array}{ll}
\begin{array}[t]{rcl}
\overline{\varphi}(x)&=&\varphi(x)\\
\overline{\varphi}(a)&=&a\\
\overline{\varphi}(\susp{\pi}{t})&=&\susp{\pi}{\overline{\varphi}(t)}
\end{array}&
\begin{array}[t]{rcl}
\overline{\varphi}([a]t)&=&[a](\overline{\varphi}(t))\\
\overline{\varphi}(t_1,\ldots,t_k)&=&(\overline{\varphi}(t_1),\ldots,
\overline{\varphi}(t_k))\\
\overline{\varphi}(f(t))&=&f(\overline{\varphi}(t)).
\end{array}
\end{array}
\end{displaymath}
Given substitutions $\varphi$ and $\gamma$ we write $\varphi\circ\gamma$ for their
composition, which is defined as follows: For every variable $x$,
$(\varphi\circ\gamma)(x)=\overline{\varphi}(t)$ where $\gamma(x)=t$. It is straightforward
to check that
$(\overline{\varphi\circ\gamma})(t) = \overline{\varphi}(\overline{\gamma}(t))$. We note
that our definition of substitution is different form those in both \cite{UPG04,CP07},
where the authors consider a function that performs the delayed permutations of the
moderated terms \emph{on-the-fly}.

\begin{lemma}[Extension to raw terms is equivariant]
  \label{lem:extension-equivariant}
  Let $\varphi$ be a substitution and $\pi$ a permutation. Then,
  $\peract{\pi}{\overline{\varphi}}= \overline{\peract{\pi}{\varphi}}$.
\end{lemma}

It is straightforward to check that the support of $\overline{\varphi}$ coincides with the
support of $\varphi$. By the above lemma, the set of extended substitutions is also a
nominal set, since it is closed with respect to the $\perm{\Atom}$-action. Hereafter we
sometimes write $\varphi(t)$, where $t$ is a raw term, instead of
$\overline{\varphi}(t)$. We may also write $\varphi^\pi$ instead of
$\peract{\pi}{\overline{\varphi}}$ or $\overline{\peract{\pi}{\varphi}}$ for
short\label{pg:varphi-pi}.

The following result highlights the relation between substitution and the permutation
action.

\begin{lemma}[Substitution and permutation action]
  \label{lem:subst-perm}
  Let $\varphi$ be a substitution, $\pi$ a permutation and $t$ a raw term.  Then,
  $\peract{\pi}{\varphi(t)} = \varphi^\pi(\peract{\pi}{t})$.
\end{lemma}

Our goal is to give meaning to ground terms in nominal sets. To this end, we need a
suitable class of algebraic structures that can be used to give an \emph{interpretation}
of those ground terms.

\begin{definition}[$\Sigma$-structure]
  \label{def:sigma-structure}
  Let $\Sigma=(\Delta,A,F)$ be a signature. A \emph{$\Sigma$-structure $M$} consists of a
  nominal set $M[\![\sigma]\!]$ for each sort $\sigma$ defined as follows
  \begin{displaymath}
    \begin{array}{rcl}
      M[\![\alpha]\!]&=&\Atom_\alpha\\
      M[\![[\alpha]\sigma]\!]&=&[\Atom_\alpha](M[\![\sigma]\!])\\
      M[\![\sigma_1\times\ldots\times\sigma_k]\!]&=&
         M[\![\sigma_1]\!]\times\ldots\times M[\![\sigma_k]\!],
    \end{array}
  \end{displaymath}
  where the $M[\![\delta_i]\!]$ with $\delta_i\in\Delta$ are given, as well as an
  equivariant function $M[\![f_{ij}]\!] : M[\![\sigma_{ij}]\!]\to M[\![\delta_i]\!]$ for
  each symbol $(f_{ij})_{\sigma_{ij}\to\delta_i}\in F$.
\end{definition}

The notion of $\Sigma$-structure adapts that of $\Sigma$-structure in \cite{CP07} to our
sorting convention with atom and abstraction sorts. The $\Sigma$-structures characterise a
range of interpretations of ground terms into elements of nominal sets, such that any sort
$\sigma$ gives rise to the expected nominal set, \ie\ atom sorts give rise to sets of
atoms, abstraction sorts give rise to sets of atom abstractions, and product sorts give
rise to finite products of nominal sets.

Next we define the \emph{interpretation of a ground term in a $\Sigma$-structure}, which
resembles the \emph{value of a term} in \cite{CP07}.

\begin{definition}[Interpretation of ground terms in a $\Sigma$-structure]
  \label{def:interpretation}
  Let $\Sigma$ be a signature and $M$ be a $\Sigma$-structure. The \emph{interpretation
    $M[\![p]\!]$ of a ground term $p$ in $M$} is given by:
  \begin{displaymath}
    \begin{array}{ll}
      \begin{array}[t]{rcl}
	M[\![a]\!]&=&a\\
	M[\![\susp{\pi}{p}]\!]&=&\peract{\pi}{M[\![p]\!]}\\
	M[\![[a]p]\!]&=& \langle a\rangle(M[\![p]\!])
      \end{array}&
      \begin{array}[t]{rcl}
        M[\![(p_1,\ldots,p_k)]\!]&=&
	(M[\![p_1]\!],\ldots,M[\![p_k]\!])\\
	M[\![f(p)]\!]&=&M[\![f]\!](M[\![p]\!]).
      \end{array}
    \end{array}
  \end{displaymath}
\end{definition}

The next lemma states that interpretation in a $\Sigma$-structure is equivariant and
highlights the relation between interpretation and moderated terms.

\begin{lemma}[Interpretation and moderated terms]
  \label{lem:interp-mod}
  Let $M$ be a $\Sigma$-structure. Interpretation in $M$ is equivariant, that is,
  $\peract{\pi}{M[\![p]\!]} = M[\![\peract{\pi}{p}]\!]$ for every ground term $p$ and
  permutation $\pi$. Moreover, $M[\![\susp{\pi}{p}]\!] = M[\![\peract{\pi}{p}]\!]$.
\end{lemma}

Finally, we introduce the $\Sigma$-structure $\mathit{NT}$, which formalises the set of
\emph{nominal terms}.

\begin{definition}[$\Sigma$-structure for nominal terms]
  \label{def:NT}
  Let $\Sigma$ be a signature. The \emph{$\Sigma$-structure $\mathit{NT}$ for nominal
    terms} is given by the least tuple $(\NT{\delta_1},\ldots,\NT{\delta_n})$ satisfying
  \begin{displaymath}
    \NT{\delta_i} = \NT{\sigma_{i1}} + \ldots + \NT{\sigma_{im_i}}\quad
    \text{for each base sort $\delta_i\in\Delta$, and}
  \end{displaymath}
  $\NT{f_{ij}}=\inj_j:\NT{\sigma_{ij}}\to\NT{\delta_i}$, for each function symbol
  $f_{ij}\in F$.

  In the conditions above, the `less than or equal to' relation for tuples is pointwise
  set inclusion. The $\NT{f_{ij}}$ is the $j$th injection of the $i$th component in
  $(\NT{\delta_1},\ldots,\NT{\delta_n})$.
\end{definition}

Nominal terms represent alpha-equivalence classes of raw terms by using the atom
abstractions of Definition~\ref{def:atom-abstraction}.
\begin{definition}[Nominal terms]
  Let $\Sigma$ be a signature. The set $\nTerm{\Sigma}_\sigma$ of \emph{nominal terms over
    $\Sigma$ of sort $\sigma$} is the domain of interpretation of the ground terms of sort
  $\sigma$ in the $\Sigma$-structure $\mathit{NT}$, that is,
  $\nTerm{\Sigma}_\sigma=\NT{\sigma}$.
\end{definition}

We sometimes write $p$, $\ell$ instead of $\NT{p}$, $\NT{\ell}$ when it is clear from the
context that we are referring to the interpretation into nominal terms of ground terms $p$
and $\ell$.

It can be checked that the nominal sets $\nTerm{\Sigma}_\sigma$ coincide (up to
isomorphism) with the nominal algebraic datatypes of Definition~8.9 in \cite{Pit13}, and
therefore by Theorem~8.15 in \cite{Pit13} the nominal terms represent alpha-equivalence
classes of raw terms.

\section{Specifications of NRTSs}
\label{sec-specification-nrts}
The NRTSs of Definition~\ref{def:nrts} are meant to be a model of computation for
calculi with name-binding operators and state/residual presentation. In this
section we present syntactic specifications for NRTSs. We start by defining nominal
residual signatures.

\begin{definition}[Nominal residual signature]
  \label{def:nominal-signature}
  A \emph{nominal residual signature} (a residual signature for short) is a quintuple
  $\Sigma=(\Delta,\ASet,\sigma,\rho,F)$ such that $(\Delta,\ASet,F)$ is a nominal
  signature and $\sigma$ and $\rho$ are distinguished nominal sorts over $\Delta$ and
  $\ASet$, which we call \emph{state sort} and \emph{residual sort} respectively. We say
  that $\nTerm{\Sigma}_\sigma$ is the set of \emph{states} and $\nTerm{\Sigma}_\rho$ is
  the set of \emph{residuals}.

  Let $\mathcal{T}=(S,R,\rel{})$ be an NRTS and $\Sigma=(\Delta,A,\sigma,\rho,F)$ be a
  residual signature. We say that $\mathcal{T}$ is an NRTS \emph{over signature} $\Sigma$
  iff the sets of states $S$ and residuals $R$ coincide with the sets of nominal terms of
  state sort $\nTerm{\Sigma}_\sigma$ and residual sort $\nTerm{\Sigma}_\rho$ respectively.
\end{definition}

Our next goal is to introduce syntactic specifications of NRTSs, which we call nominal
residual transition system specifications. To this end, we will make use of residual
formulas and freshness assertions over raw terms, which are defined below.

\begin{definition}[Residual formula and freshness assertion]
  A \emph{residual formula} (a \emph{formula} for short) over a residual signature
  $\Sigma$ is a pair $(s,r)$, where $s\in\rTerm{\Sigma}_\sigma$ and
  $r\in\rTerm{\Sigma}_\rho$. We use the more suggestive $s\rel{}r$ in lieu of $(s,r)$. A
  formula $s\rel{}r$ is \emph{ground} iff $s$ and $r$ are ground terms.

  A \emph{freshness assertion} (an \emph{assertion} for short) over a signature $\Sigma$
  is a pair $(a,t)$ where $a\in\Atom$ and $t\in\rTerm{\Sigma}$. We will write $a\fra t$ in
  lieu of $(a,t)$.  An assertion is \emph{ground} iff $t$ is a ground term.
\end{definition}
\begin{remark}
  Formulas and assertions are raw syntactic objects, similar to raw terms, which will
  occur in the rules of the nominal residual transition system specifications to be
  defined, and whose purpose is to represent respectively transitions and freshness
  relations involving nominal terms. A formula $s\rel{}r$ (resp. an assertion $a\fra t$)
  unifies with a ground formula $\varphi(s)\rel{}\varphi(r)$ (resp. a ground assertion
  $a\fra \varphi(t)$), which in turn represents a transition
  $\NT{\varphi(s)}\rel{}\NT{\varphi(r)}$ (resp. a freshness relation
  $a\# \NT{\varphi(t)}$).  For the assertions, notice how the symbols $\fra$, $\#$ and
  $\NT{\ }$ interact. The ground assertion $a\fra [a]a$ represents the freshness relation
  $a\#\NT{[a]a}$, which is true. On the other hand, the freshness relation $a\#[a]a$ is
  false because $a\in\supp([a]a)$.
\end{remark}

Permutation action and substitution extend to formulas and assertions in the expected
way. Formulas and assertions are elements of nominal sets. Their support is the union of the
supports of the raw terms in them, hence we write ${\supp(t\rel{}t')}$ and
${\supp(a\fra t)}$.  We will also write $b\#(t\rel{} t')$ and $b\#(a\fra t)$ for freshness
relations involving formulas and assertions respectively.

\begin{definition}[Nominal residual transition system specification]
  \label{def:NRTSS}
  Let $\Sigma$ be a residual signature $(\Delta,\ASet,\sigma,\rho,F)$. A \emph{transition
    rule} \emph{over} $\Sigma$ (a \emph{rule}, for short) is of the form
  \begin{mathpar}
    \inferrule*[right={,}]
    {\{u_i\rel{}u'_i\mid i\in I\}\quad \{a_j\fra v_j\mid j\in J\}}
    {t\rel{}t'}
  \end{mathpar}
  abbreviated as $H,\nabla/t\rel{}t'$, where $H={\{u_i\rel{}u'_i\mid i\in I\}}$ is a
  finitely supported set of formulas over $\Sigma$ (we call $H$ the set of
  \emph{premisses}) and where $\nabla=\{a_j\fra v_j\mid j\in J\}$ is a finite set of
  assertions over $\Sigma$ (we call $\nabla$ the \emph{freshness environment}). We say
  formula $t\rel{}t'$ over $\Sigma$ is the \emph{conclusion}, where $t$ is the
  \emph{source} and $t'$ is the \emph{target}. A rule is an \emph{axiom} iff it has an
  empty set of premisses. Note that axioms might have a non-empty freshness environment.

  A \emph{nominal residual transition system specification over} $\Sigma$ (abbreviated to
  NRTSS) is a set of transition rules over $\Sigma$.
\end{definition}

Permutation action and substitution extend to rules in the expected way; they are applied
to each of the formulas and freshness assertions in the rule.

Notice that the rules of an NRTSS are elements of a nominal set. The support of a rule
$H,\nabla/t\rel{}t'$ is the union of the support of $H$, the support of $\nabla$ and the
support of $t\rel{}t'$. %
In the sequel we write $\supp(\textsc{Ru})$ for the support of rule \textsc{Ru}, and
$a\#\textsc{Ru}$ for a freshness relation involving atom $a$ and rule \textsc{Ru}. Observe
that the set $H$ of premisses of a rule may be infinite, but its support must be
finite. However, the freshness environment $\nabla$ must be finite in order to make the
simplification rules of Definition~\ref{def:simplification} to follow terminating. These
simplification rules will be used in Section~\ref{sec-rule-format-nrts} to define the rule
format in Definition~\ref{def:alpha-conv-format}.

Let $\R$ be an NRTSS. We say that the formula $s\rel{}r$ \emph{unifies} with rule \Ru\ in
$\R$ iff \Ru\ has conclusion $t\rel{}t'$ and $s\rel{}r$ is a substitution instance of
$t\rel{}t'$. If $s$ and $r$ are ground terms, we also say that transition
$\NT{s}\rel{}\NT{r}$ unifies with \Ru.

\begin{definition}[Proof tree]
  \label{def:proof-tree}
  Let $\Sigma$ be a residual signature and $\R$ be an NRTSS over $\Sigma$. A proof tree in
  $\R$ of a transition $\NT{s}\rel{}\NT{r}$ is an upwardly branching rooted tree without
  paths of infinite length whose nodes are labelled by transitions such that
  \begin{enumerate}[(i)]
  \item the root is labelled by $\NT{s}\rel{}\NT{r}$, and
  \item if $K=\{\NT{q_i}\rel{}\NT{q_i'}\mid i\in I\}$ is the set of labels of the nodes
    directly above a node with label ${\NT{p}\rel{}\NT{p'}}$, then there exist a rule
    \begin{mathpar}
      \inferrule*
      {\{u_i\rel{}u'_i\mid i\in I\} \qquad
        \{a_j\fra v_j\mid j\in J\}}
      {t\rel{}t'}
    \end{mathpar}
    in $\R$ and a ground substitution $\varphi$ such that $\varphi(t\rel{}t')=p\rel{}p'$
    and, for each $i\in I$ and for each $j\in J$, $\varphi(u_i\rel{}u_i')=q_i\rel{}q_i'$
    and $a_j\#\NT{\varphi(v_j)}$ hold.
  \end{enumerate}
  We say that $\NT{s}\rel{}\NT{r}$ is \emph{provable} in $\R$ iff it has a proof
  tree in $\R$. The transition relation specified by $\R$ consists of all the
  transitions that are provable in $\R$.
\end{definition}

The nodes of a proof tree are labelled by transitions, which contain nominal terms (\ie\
syntactic objects that use the atom abstractions of
Definition~\ref{def:atom-abstraction}). The use of nominal terms captures the convention
in typical nominal calculi of considering terms `up to alpha-equivalence'.

The fact that the nodes of a proof tree are labelled by nominal terms is the main
difference between our approach and previous work in nominal structural operational
semantics \cite{ACGIMR}, nominal rewriting \cite{UPG04,FG07} and nominal algebra
\cite{GM09}. In all these works, the `up-to-alpha-equivalence' transitions are explicitly
instrumented within the model of computation by adding to the specification system
inference rules that perform alpha-conversion of raw terms.

\section{Rule formats for NRTSSs}
\label{sec-rule-format-nrts}
This section defines two rule formats for NRTSSs that ensure that:
\begin{inparaenum}[(i)]
\item an NRTSS induces an equivariant transition relation, and thus an NRTS of
  Definition~\ref{def:nrts};
\item an NRTSS induces a transition relation which, together with an equivariant function
  $\bn$, corresponds to an NTS of Definition~\ref{def:nts} \cite{PBEGW15}. For the latter,
  we need to ensure that the induced transition relation is equivariant and satisfies
  \emph{alpha-conversion of residuals} (recall, if $p\rel{}(\ell,p')$ is provable in $\R$
  and $a$ is in the set of binding names of $\ell$, then for every atom $b$ that is fresh
  in $(\ell,p')$ the transition $p\rel{}\peract{\tr{a}{b}}{(\ell,p')}$ is also provable).
\end{inparaenum}

As a first step, we introduce a rule format ensuring equivariance of the induced
transition relation.

\begin{definition}[Equivariant format]
  \label{def:equivariant-format}
  Let $\mathcal{R}$ be an NRTSS. $\mathcal{R}$ is in \emph{equivariant format} iff the
  rule $\peract{\tr{a}{b}}{\textsc{Ru}}$ is in $\mathcal{R}$, for every rule $\textsc{Ru}$
  in $\R$ and for each $a,b\in\Atom$.
\end{definition}

\begin{lemma}
  \label{lem:equivariant-format}
  Let $\mathcal{R}$ be an NRTSS in equivariant format. For every rule $\textsc{Ru}$ in
  $\mathcal{R}$ and for every permutation $\pi$, the rule $\peract{\pi}{\textsc{Ru}}$ is
  in $\mathcal{R}$.
\end{lemma}

\begin{theorem}[Rule format for NRTSs]
  \label{thm:rule-format-equivariance}
  Let $\R$ be an NRTSS. If $\R$ is in equivariant format then $\R$ induces an NRTS.
\end{theorem}

Before introducing a rule format ensuring alpha-conversion of residuals, we adapt to our
freshness environments the simplification rules and the entailment relation of
Definition~10 and Lemma~15 in \cite{FG07}, which we will use in the definition of the rule
format.

\begin{definition}[Simplification of freshness environments]
  \label{def:simplification}
  Consider a signature $\Sigma$. The following rules, where $a$, $b$ are assumed to be
  distinct atoms and $\nabla$ is a freshness environment over $\Sigma$, define
  \emph{simplification of freshness environments}:
  \begin{displaymath}
    \begin{array}{l@{\hspace*{-.9cm}}l}\hspace*{-.3cm}
      \begin{array}[t]{rcl}
	\{a\fra b\}\cup \nabla&\Longrightarrow&\nabla\\
	\{a\fra \susp{\pi}{t}\}\cup \nabla&\Longrightarrow&
           \{\peract{\pi^{-1}}{a}\fra t\}\cup \nabla\\
	\{a\fra [b]p\}\cup\nabla&\Longrightarrow&\{a\fra p\}\cup\nabla
      \end{array} &
      \begin{array}[t]{rcl}
        \{a\fra (p_1,\ldots,p_k)\}\cup\nabla&\Longrightarrow&
            \{a\fra p_i,\ldots,a\fra p_k\}\cup\nabla\\
        \{a\fra [a]p\}\cup\nabla&\Longrightarrow&\nabla\\
        \{a\fra f(p)\}\cup\nabla&\Longrightarrow&\{a\fra p\}\cup\nabla.
      \end{array}
    \end{array}
  \end{displaymath}
  The rules define a reduction relation on freshness environments. We write
  $\nabla\Longrightarrow\nabla'$ when $\nabla'$ is obtained from $\nabla$ by applying one
  simplification rule, and $\Longrightarrow^*$ for the reflexive and transitive closure of
  $\Longrightarrow$.
\end{definition}

\begin{lemma}
  \label{lem:unique-nf}
  The relation $\Longrightarrow$ is confluent and terminating.
\end{lemma}

A freshness assertion is \emph{reduced} iff it is of the form $a\fra a$ or $a \fra x$. We
say that $a\fra a$ is \emph{inconsistent} and $a \fra x$ is \emph{consistent}. An
environment $\nabla$ is \emph{reduced} iff it consists only of reduced assertions. An
environment containing a freshness assertion that is not reduced can always be simplified
using one of the rules in Definition~\ref{def:simplification}. Therefore, by
Lemma~\ref{lem:unique-nf}, an environment $\nabla$ reduces by $\Longrightarrow^*$ to a
unique reduced environment, which we call the \emph{normal form} of $\nabla$, written
$\nf{\nabla}$. An environment $\nabla$ is \emph{inconsistent} iff $\nf{\nabla}$ contains
some inconsistent assertion. We say $\nabla$ \emph{entails} $\nabla'$ (written
$\nabla \vdash \nabla'$) iff either $\nabla$ is an inconsistent environment, or
$\nf{\nabla'}\subseteq\nf{\nabla}$. We write $\vdash\nabla$ iff $\emptyset \vdash \nabla$.

\begin{lemma}
  \label{lem:entails}
  Let $\nabla$ be an environment over $\Sigma$. Then, for every ground substitution
  $\varphi$, the conjunction of the freshness relations represented by
  $\varphi(\nf{\nabla})$ holds iff the conjunction of the freshness relations represented
  by $\varphi(\nabla)$ hold.
\end{lemma}

In particular, if $\vdash \nabla$ then for every ground substitution $\varphi$ the
freshness relations represented by $\varphi(\nabla)$ hold.

We are interested in NTS \cite{PBEGW15}, which consider signatures with base sorts $\Act$
and $\Proc$, with a single atom sort $\Chan$ and with source and residual sorts $\Proc$
and $\Act\times\Proc$ respectively. We let $\SigmaNTS$ be any such signature parametric on
a set $F$ of function symbols that we keep implicit. We let
$\bn:\nTerm{\Sigma}_\Act\to\pset_\omega(\Atom_\Chan)$ be the binding-names function of a
given NTS. From now on we restrict the attention to the NTS of \cite{PBEGW15}, and the
definitions and results to come apply to NRTS/NRTSS over a signature $\SigmaNTS$. We
require that the rules of an NRTSS only contain ground actions $\ell$ and therefore
function $\bn$ is always defined over $\NT{\ell}$. (Recall that we write $\bn(\ell)$
instead of $\bn(\NT{\ell})$ since it is clear in this context that the $\ell$ stands for a
nominal term.) The rule format that we introduce in Definition~\ref{def:alpha-conv-format}
relies on identifying the rules that give rise to transitions with actions $\ell$ such
that $\bn(\ell)$ is non-empty. To this end, we adapt the notion of strict stratification
from \cite{AFGI17,FV03}.

\begin{definition}[Partial strict stratification]
  \label{def:partial-strict stratification}
  Let $\mathcal{R}$ be an NRTSS over a signature $\SigmaNTS$ and $\bn$ be a binding-names
  function. Let $S$ be a partial map from pairs of ground processes and actions to ordinal
  numbers. $S$ is a \emph{partial strict stratification of $\mathcal{R}$ with respect to
    $\bn$} iff
  \begin{enumerate}[(i)]
  \item $S(\varphi(t), \ell) \not= \bot$, for every rule in $\R$ with conclusion
    $t \rel{} (\ell, t')$ such that $\bn(\ell)$ is non-empty and for every ground
    substitution $\varphi$, and
  \item $S(\varphi(u_i),\ell_i)<S(\varphi(t),\ell)$ and $S(\varphi(u_i),\ell_i)\not=\bot$,
    for every rule in $\R$ with conclusion $t\rel{}(\ell,t')$ such that
    $S(\varphi(t),\ell)\not=\bot$, for every premiss $u_i\rel{}(\ell_i,u'_i)$ of $\R$ and
    for every ground substitution $\varphi$.
  \end{enumerate}
  We say a pair $(p,\ell)$ of ground process and action \emph{has order} $S(p,\ell)$.
\end{definition}
The choice of $S$ determines which rules will be considered by the rule format for NRTSSs
of Definition~\ref{def:alpha-conv-format} below, which guarantees that the induced
transition relation satisfies alpha-conversion of residuals and, therefore, the associated
transition relation together with function $\bn$ are indeed an NTS.
We will intend the map $S$ to be such that the only rules whose source and label of the
conclusion have defined order are those that may take part in proof trees of transitions
with some binding atom in the action.

\begin{definition}[Alpha-conversion-of-residuals format]
  \label{def:alpha-conv-format}
  Let $\mathcal{R}$ be an NRTSS over a signature $\SigmaNTS$, $\bn$ be a binding-names
  function and $S$ be a partial strict stratification of $\mathcal{R}$ with respect to
  $\bn$. Assume that all the actions occurring in the rules of $\mathcal{R}$ are
  ground. Let
  \begin{mathpar}
    \inferrule*[right=Ru] {\{u_i\rel{}(\ell_i,u'_i)\mid i\in I\} \qquad\nabla}
    {t\rel{}(\ell,t')}
  \end{mathpar}
  be a rule in $\mathcal{R}$. Let $D$ be the set of variables that occur in the source $t$
  of \textsc{Ru} but do not occur in the premisses $u_i\rel{}(\ell_i,u'_i)$ with $i\in I$,
  the environment $\nabla$ or the target $t'$ of the rule. The rule \textsc{Ru} is in
  \emph{alpha-conversion-of-residuals format with respect to $S$} (\emph{ACR format with
    respect to $S$} for short) iff for each ground substitution $\varphi$ such that
  $S(\varphi(t),\ell)\neq\bot$, there exists a ground substitution $\gamma$ such that
  $\dom(\gamma)\subseteq D$, and for every atom $a$ in the set
  $\Atom\setminus \{c\in\supp(t)\mid \nf{\{c\fra t\}}=\emptyset\}$ and for every atom
  $b\in\bn(\ell)$, the following hold:
  \begin{enumerate}[(i)]
  \item\label{it:one} $\{a\fra t'\} \cup \nabla \vdash \{a\fra u'_i \mid i\in I\}$,
  \item\label{it:two}
    $\{a\fra t'\} \cup \nabla \cup \{a\fra u_i \mid i\in I\} \vdash \{a\fra \gamma(t)\}$,
    and
  \item\label{it:three}
    $\nabla\cup\{b\fra u_i\mid i\in I \land b\in \bn(\ell_i)\} \vdash \{b\fra\gamma(t)\}$.
  \end{enumerate}
  An NRTSS $\mathcal{R}$, together with a binding-names function $\bn$ is in \emph{ACR
    format} iff $\mathcal{R}$ is in equivariant format and there exists a partial strict
  stratification $S$ such that all the rules in $\mathcal{R}$ are in ACR format with
  respect to $S$.
\end{definition}

Given a transition $p\rel{}(\ell,q)$ that unifies with the conclusion of \Ru, the rule
format ensures that any atom $a$ fresh in $(\ell, q)$ is also fresh in $p$, and also that
the binding atom $b$ is fresh in $p$. We have obtained the constraints of the rule format
by considering the variable flow in each node of a proof tree and the freshness relations
that we want to ensure. Constraints~(i) and (ii) cover the case for the freshness relation
$a\#p$ and Constraint~(iii) covers the case for the freshness relation $b\#p$. The purpose
of substitution $\gamma$ is to ignore the variables that occur in the source of a rule but
are dropped everywhere else in the rule.
Constraints~(i) and (ii) are not required for
atoms $a$ that for sure are fresh in $p$, and this explains why the $a$ in the rule format
ranges over $\Atom\setminus \{c\in\supp(t)\mid \nf{\{c\fra t\}}=\emptyset\}$.  For
instance, take rule \textsc{ResB} from Section~\ref{sec:early-pi-calculus}. Condition
$\{c\fra (\boutAA[a]{b},\newPA[{[c]y}]),$
{$c\fra \boutAA[a]{b}\} \vdash\{c\fra (\boutAA[a]{b},y)\}$} does not hold because
$c\fra [c]y$ does not entail that $c\fra y$. However, for a transition
$\NT{\newPA[{[c]p}]}\rel{} \NT{(\ell,\newPA[{[c]p'}])}$, $c$ is fresh in
$\NT{\newPA[{[c]p}]}$ even if $c$ is not fresh in $\NT{p}$. %
\begin{theorem}[Rule format for NTSs]
  \label{the:alpha-conversion}
  Let $\mathcal{R}$ be an NRTSS. If $\mathcal{R}$, together with the binding-names
  function $\bn$, is in ACR format then the NRTS induced by $\mathcal{R}$ and $\bn$
  constitute an NTS---that is, the transition relation induced by $\mathcal{R}$ is
  equivariant and satisfies alpha-conversion of residuals.
\end{theorem}

\begin{proof}[Sketch of the proof]
  Given a transition $\NT{\varphi(t)}\rel{}\NT{\varphi((\ell,t'))}$, we first prove the
  freshness relations $a\#\NT{\varphi(\gamma(t))}$ and $b\#\NT{\varphi(\gamma(t))}$. Both
  relations are proven by induction on $S(\varphi(\gamma(t)),\ell)$, and by analysing the
  variable flow in the rule unifying with $\varphi(t)\rel{}\varphi(\ell,t')$. For the
  first relation, we assume $a\#\NT{\varphi(t')}$, use Constraint~(i) to prove that
  $a\#\NT{\varphi(u'_i)}$ for each target $u'_i$ of a premiss, apply the induction
  hypothesis to obtain $a\#\NT{\varphi(\gamma(u_i)}$ for each source of a premiss $u_i$,
  and use Constraint~(ii) to conclude that $a\#\NT{\varphi(\gamma(t))}$. For the second
  relation, the induction hypothesis ensures that $b\#\NT{\varphi(\gamma(u_i))}$ for each
  source $u_i$ of a premiss having $b$ as a binding name, and we use Constraint~(iii) to
  conclude that $b\#\NT{\varphi(\gamma(t))}$. From these two freshness relations it is
  straightforward to prove that
  $\NT{\varphi(t)}\rel{}\peract{\tr{a}{b}}{\NT{\varphi((\ell,t'))}}$ and we are done.
\end{proof}

\section{Example of application to the early $\pi$-calculus}
\label{sec:early-pi-calculus}
Consider the NRTSS $\mathcal{R}$ for the early $\pi$-calculus \cite{MPW92} over a
signature $\SigmaNTS$ where $F$ is the set of function symbols from
Example~\ref{ex:pi-calculus}.  Below we collect an excerpt of the rules, where
$a,b,c\in \Atom_\Chan$ and $\ell$ is a ground action:
\begin{mathpar}
  \inferrule*[right=In]
  {b\fra [c]x}
  {\inPA[a]{[c]x}\rel{}\resAF{\inAA[a]{b}}{\susp{\tr{c}{b}}{x}}}
  \and
  \inferrule*[right=Open]
  {x\rel{}\resAF{\outAA[a]{b}}{y}\qquad b\fra a}
  {\newPA[{[b]x}]\rel{}\resAF{\boutAA[a]{b}}{y}}
  \\
  \inferrule*[right=SumL]
  {x_1\rel{}\resAF{\ell}{y_1}}
  {\sumPA[x_1]{x_2} \rel{}\resAF{\ell}{y_1}}
  \and
  \inferrule*[right=ParL,left=$\ell\not\in\{\mathit{boutA}(a{,}\,b)\}$]
  {x_1\rel{}\resAF{\ell}{y_1}}
  {\parPA[x_1]{x_2} \rel{}\resAF{\ell}{(\parPA[y_1]{x_2})}}
  \\
  \inferrule*[right=ParResL]
  {x_1 \rel{} (\boutAA[a]{b},y_1)\qquad b\fra x_2}
  {\parPA[x_1]{x_2} \rel{} (\boutAA[a]{b},(\parPA[y_1]{x_2}))}
  \\
  \inferrule*[right=Out]
  { }
  {\outPA[a]{b}{x}\rel{}\resAF{\outAA[a]{b}}{x}}
  \and
  \inferrule*[right=Rep]
  {x\rel{}\resAF{\ell}{y}}
  {\repPA[x] \rel{}\resAF{\ell}{(\parPA[y]{\repPA[x]})}}
  \\
  \inferrule*[right=CloseL]
  {x_1\rel{}\resAF{\boutAA[a]{b}}{y_1}
    \qquad x_2\rel{}\resAF{\inAA[a]{b}}{y_2} \qquad b\fra x_2}
  {\parPA[x_1]{x_2}\rel{}\resAF{\tauAA}{\newPA[{[b](\parPA[y_1]{y_2})}]}}
\end{mathpar}

An input process $\NT{\inPA[a]{[c]p}}$ can perform a transition to a process
$\NT{\peract{\tr{c}{b}}{p}}$ that is obtained by substituting a channel name $b$ received
through channel $a$ for channel name $c$ in $p$. In the rule \textsc{In}, the moderated
term $\susp{\tr{c}{b}}{x}$ needs to be used in order to indicate that permutation
$\tr{c}{b}$ will be performed over the term substituted for variable $x$.

The rule \textsc{CloseL} specifies the interaction of a process such as
$\NT{\newPA[{[b](\outPA[a]{b}{p})}]}$, which exports a private channel name $b$ through
channel $a$, composed in parallel with an input process such as $\NT{\inPA[a]{[c]q}}$ that
reads through channel $a$. The private name $b$ is exported and the resulting process
$\NT{\newPA[{[b](\parPA[p]{\peract{\tr{c}{b}}{q}})}]}$ is the parallel composition of
processes $p$ and $q$ where atom $b$ is restricted. For illustration, consider the raw
terms $t\equiv \newPA[{[b](\outPA[a]{b}{p})}]$ and $t'\equiv(\boutAA[a]{b},p)$. The
transition $\NT{t}\rel{}\NT{t'}$ is provable in $\mathcal{R}$ by the following proof tree:
\begin{mathpar}
  \inferrule*[right=Open{\normalsize$,\ \textup{as}\ b\#a$}.]
  {
    \inferrule*[right=Out]
    { }
    {\NT{\outPA[a]{b}{p}}\rel{}\NT{(\outAA[a]{b},p)}}
  }
  {\NT{\newPA[{[b](\outPA[a]{b}{p})}]}\rel{}\NT{(\boutAA[a]{b},p)}}
\end{mathpar}

Notice that the nodes of the proof tree above are labelled by transitions involving
nominal terms. Therefore, if we were to start with the raw term
$q\equiv\newPA[{[c](\outPA[a]{c}{p})}]$---which is alpha-equivalent to $t$---then the
transition $\NT{q}\rel{}\NT{t'}$ would have the same proof tree as above, since $\NT{t}$
and $\NT{q}$ are the same nominal term. This contrasts with the related work
\cite{CMRG12,FG07}, which considers raw terms in the model of computation and
instruments alpha-conversion explicitly in the specification system.

We use the rule format of Definition~\ref{def:alpha-conv-format} to show that
$\mathcal{R}$, together with equivariant function
$\bn(\ell)=\{b\mid \ell = \boutAA[a]{b}\}$ specifies an NTS. We consider the following
partial strict stratification:
\begin{displaymath}
\begin{array}{rcl}
S(\outPA[a]{b}{p},\outAA[a]{b})&=&0\\
S(\parPA[p]{q},\ell)&=&1+\max\{S(p,\ell),S(q,\ell)\}
\quad \ell\in\{\boutAA[a]{b},\outAA[a]{b}\}\\
S(\sumPA[p]{q},\ell)&=&1+\max\{S(p,\ell),S(q,\ell)\}
\quad \ell\in\{\boutAA[a]{b},\outAA[a]{b}\}\\
S(\repPA[p],\ell)&=&1+S(p,\ell)
\quad \ell\in\{\boutAA[a]{b},\outAA[a]{b}\}\\
S(\newPA[{[c]p}],\ell)&=&1+S(p,\ell)
\quad \ell\in\{\boutAA[a]{b},\outAA[a]{b}\}\ \text{and}\ c \not\in\{a,b\}\\
S(\newPA[{[b]p}],\boutAA[a]{b})&=&1+S(p,\outAA[a]{b})\\
S(p,\ell)&=&\bot\quad \text{o.w.}
\end{array}
\end{displaymath}

We check that $\mathcal{R}$ is in ACR format as follows. The only rules in $\mathcal{R}$
whose sources and actions unify with pairs of processes and actions that have defined
order are \textsc{Out}, \textsc{Open} and \textsc{ParResL}, and the instance of rule
\textsc{ParL} where $\ell=\outAA[a]{b}$, and the instances of rules \textsc{SumL},
\textsc{Rep} and \textsc{Res} where $\ell\in\{\boutAA[a]{b},\outAA[a]{b}\}$ (and the
corresponding instances of the symmetric versions \textsc{ParResR}, \textsc{ParR} and
\textsc{SumR}, which are omitted in the excerpt). We will only check the ACR-format for
rules \textsc{Out}, \textsc{SumL} and \textsc{Open}.

For rule \textsc{Out}, we have an empty set of premisses and the set $D$ of atoms that are
in $\supp(\outPA[a]{b}{x})$ but are not in $\supp(\outAA[a]{b},x)$ is empty. Therefore we
can do away with substitution~$\gamma$. There is no atom $a$ such that
$\nf{\{a\fra \outPA[a]{b}{x}\}=\emptyset}$ and the set $\bn(\outAA[a]{b})$ is empty. We
only need to check that for every atom $c$,
$\{c\fra (\outAA[a]{b},x)\}\vdash\{c\fra \outPA[a]{b}{x}\}$. For atoms
$c\in\supp(\outAA[a]{b},x)$ the obligation of the rule format vacuously holds, and
therefore it is enough to pick an atom $c$ fresh in the rule and check that
$\{c\fra (\outAA[a]{b},x)\}\vdash\{c\fra \outPA[a]{b}{x}\}$, which is straightforward.

For rule \textsc{SumL}, we first check the instance where $\ell=\boutAA[a]{b}$. We have
premiss $x_1\rel{}(\boutAA[a]{b},y_1)$ and the set $D$ contains $x_2$. We pick $\gamma$
such that $\gamma(x_2)=\nullPA$. There is no atom $a$ such that
$\nf{\{a\fra \sumPA[x_1]{x_2}\}}=\emptyset$ and the set $\bn(\boutAA[a]{b})$ contains atom
$b$. Again, it is enough to pick atom $c$ fresh in the rule and check that
\begin{displaymath}
\begin{array}{c}
\{c\fra (\boutAA[a]{b},y_1)\}
\vdash\{c\fra (\boutAA[a]{b},y_1)\}\qquad\text{and}\\
\{c\fra (\boutAA[a]{b},y_1), c\fra x_1\}\vdash\{c\fra \gamma(\sumPA[x_1]{x_2})\}
\qquad\text{and}\\
\{b\fra x_1\}\vdash\{b\fra \gamma(\sumPA[x_1]{x_2})\},
\end{array}
\end{displaymath}
which holds since $\gamma(\sumPA[x_1]{x_2})=\sumPA[x_1]{\nullPA}$ and $b\fra \nullPA$
reduces to the empty set.

Now we check the instance where $\ell=\outAA[a]{b}$. We have premiss
$x_1\rel{}(\outAA[a]{b},y_1)$ and the set $D$ and the substitution $\gamma$ are the same
as before. There is no atom $a$ such that $\nf{\{a\fra \sumPA[x_1]{x_2}\}}=\emptyset$ and
the set $\bn(\outAA[a]{b})$ is empty. Again, it is enough to pick atom $c$ fresh in the
rule and check that $\{c\fra (\outAA[a]{b},y_1)\} \vdash\{c\fra (\outAA[a]{b},y_1)\}$ and
$\{c\fra (\outAA[a]{b},y_1), c\fra x_1\}\vdash\{c\fra \gamma(\sumPA[x_1]{x_2})\}$, which
holds as before.

For rule \textsc{Open} the set $D$ is empty and $\nf{\{b\fra \newPA[{[b]x}]\}}=\emptyset$.
It is enough to pick atom $c$ fresh in the rule (and therefore different from $b$) and
check that
\begin{displaymath}
\begin{array}{c}
\{c\fra (\boutAA[a]{b},y),b\fra a\}
\vdash\{c\fra (\boutAA[a]{b},y)\}\qquad\text{and}\\
\{c\fra (\boutAA[a]{b},y), b\fra a, c\fra x\}
\vdash\{c\fra \newPA[{[b]x}])\}
\qquad\text{and}\\
\{b\fra x, b \fra a\}\vdash\{b\fra \newPA[{[b]x}]\},
\end{array}
\end{displaymath}
which holds because $b\fra \newPA[{[b]x}]$ reduces to the empty set.

Atoms $a$, $b$ and $c$ in the specification of $\mathcal{R}$ range over $\Atom_\Chan$, and
thus $\mathcal{R}$ is in equivariant format. Therefore $\R$ is in ARC format. By
Theorem~\ref{the:alpha-conversion} the NRTS induced by $\mathcal{R}$, together with
function $\bn$, constitute an NTS of Definition~\ref{def:nts}.

\section{Conclusions and future work}
\label{sec:conclusions}

The work we have presented in this paper stems from the recently proposed Nominal SOS
(NoSOS) framework~\cite{CMRG12} and from earlier proposals for nominal logic
in~\cite{UPG04,CP07,GM09}. It is by no means the only approach studied so far in the
literature that aims at a uniform treatment of binders and names in programming and
specification languages. Other existing approaches that accommodate variables and binders
within the SOS framework are those proposed by Fokkink and Verhoef in~\cite{FV98}, by
Middelburg in~\cite{Mid01,Mid03}, by Bernstein in~\cite{Ber98}, by Ziegler, Miller and
Palamidessi in~\cite{ZMP06} and by Fiore and Staton in~\cite{FS09} (originally, by Fiore
and Turi in~\cite{FT01}). The aim of all of the above-mentioned frameworks is to establish
sufficient syntactic conditions guaranteeing the validity of a semantic result (congruence
in the case of~\cite{Ber98,FS09,Mid01,ZMP06} and conservativity in the case
of~\cite{FV98,Mid03}). In addition, Gabbay and Mathijssen present a nominal axiomatisation
of the $\lambda$-calculus in~\cite{GM10}. None of these approaches addresses equivariance
nor the property of alpha-conversion of residuals in \cite{PBEGW15}.

Our current proposal aims at following closely the spirit of the seminal work on nominal
techniques by Gabbay, Pitts and their co-workers, and paves the way for the development of
results on rule formats akin to those presented in the aforementioned references.  Amongst
those, we consider the development of a congruence format for the notion of bisimilarity
presented in~\cite[Def.~2]{PBEGW15} to be of particular interest. The logical
characterisation of bisimilarity given in~\cite{PBEGW15} opens the intriguing possibility
of employing the divide-and-congruence approach from~\cite{FGW06} to obtain an elegant
congruence format and a compositional proof system for the logic.

In the NTSs of Parrow {\etal}~\cite{PBEGW15}, scope opening is modelled by the property of
alpha-conversion of residuals. We are currently exploring an alternative in which scope
opening is encoded by a \emph{residual abstraction} of sort $[\Chan](\Act\times\Proc)$.
We have developed mutual, one-to-one translations between the NTSs and the NRTSs with
residual abstractions. The generality of our NRTSs also allows for neat specifications of
variants of the $\pi$-calculus such as Sangiorgi's internal $\pi$-calculus~\cite{San96}.

Developing rule formats for SOS is always the result of a trade-off between ease of
application and generality. Our rule format for alpha-conversion of residuals in
Definition~\ref{def:alpha-conv-format} is no exception and might be generalised in various
ways. For instance, the quantification on atom $a$ in conditions (i) and (ii), and the use
of substitution $\gamma$ might be made more general by a finer analysis of the variable
flow in a rule.  Another generalisation of the rule format would consider possibly open
raw actions.

Finally, we are developing rule formats for properties other than alpha-conversion of
residuals. One such rule format ensures a \emph{non-dropping property} for NRTSs to the
effect that, in each transition, the support of a state is a subset of the support of its
derivative.


\end{document}